\documentclass[sigconf]{acmart-edbt2018}

\renewcommand\footnotetextcopyrightpermission[1]{} % removes footnote with conference information in first column
\usepackage{booktabs} % For formal tables

\usepackage[usenames,dvipsnames]{pstricks}
\usepackage{epsfig}
\usepackage{fancyhdr}
\usepackage{datetime}
\usepackage{framed}
\usepackage{amsmath}
\usepackage{float} 
\usepackage{makecell}
\makeatletter
\newif\if@restonecol
\makeatother

\usepackage[ruled,vlined,norelsize,linesnumbered]{algorithm2e}
\usepackage{amssymb}
\usepackage{comment}
\usepackage{mathtools}
\usepackage{amsthm}

\newcommand{\todo}[1]{{\color{red}{\bf #1}}}
\newcommand{\com}[1]{{\color{blue}{\bf \it #1}}}
\newcommand{\spara}[1]{\smallskip\noindent{\bf{#1}}}

\newcommand{\moreSpacing}{\renewcommand{\baselinestretch}{1.11} \normalsize}
\newcommand{\smallExtraSpacing}{\renewcommand{\baselinestretch}{1.02} \normalsize}

\makeatletter
\def\@copyrightspace{\relax}
\makeatother

\makeatletter
\renewcommand\@formatdoi[1]{\ignorespaces}
\makeatother\makeatletter
\renewcommand\@formatdoi[1]{\ignorespaces}

\makeatletter
\def\@copyrightspace{\relax}
\makeatother

\begin{document}

\title{Core Discovery in Hidden Graphs}
%\titlenote{Produces the permission block, and copyright information}
%\subtitle{Extended Abstract}
%\subtitlenote{The full version of the author's guide is available as
%  \texttt{acmart.pdf} document}

\author{Panagiotis Strouthopoulos}
%\authornote{Dr.~Trovato insisted his name be first.}
\affiliation{%
  \institution{Aristotle University of Thessaloniki}
  %\streetaddress{P.O. Box 1212}
  \city{Thessaloniki} 
  \state{Greece} 
  \postcode{GR54124}
}
\email{pstrouth@csd.auth.gr}

\author{Apostolos N. Papadopoulos}
%\authornote{The secretary disavows any knowledge of this author's actions.}
\affiliation{%
  \institution{Aristotle University of Thessaloniki}
  %\streetaddress{P.O. Box 1212}
  \city{Thessaloniki} 
  \state{Greece} 
  \postcode{GR54124}
}
\email{papadopo@csd.auth.gr}
\vspace{1cm}

% \renewcommand{\shortauthors}{B. Trovato et al.}
%\renewcommand{\shortauthors}{}

\begin{abstract} 
Network exploration is an important research direction with many applications. In such a setting, the network is, usually, modeled as a graph $G$, whereas any structural information of interest is extracted by inspecting the way nodes are connected together. In the case where the adjacency matrix or the adjacency list of $G$ is available, one can directly apply graph mining algorithms to extract useful knowledge. However, there are cases where this is not possible because the graph is \textit{hidden} or \textit{implicit}, meaning that the edges are not recorded explicitly in the form of an adjacency representation. In such a case, the only alternative is to pose a sequence of \textit{edge probing queries} asking for the existence or not of a particular graph edge. However, checking all possible node pairs is costly (quadratic on the number of nodes). Thus, our objective is to pose as few edge probing queries as possible, since each such query is expected to be costly. In this work, we center our focus on the \textit{core decomposition} of a hidden graph. In particular, we provide an efficient algorithm to detect the maximal subgraph of $S_k$ of $G$ where the induced degree of every node $u \in S_k$ is at least $k$. Performance evaluation results demonstrate that significant performance improvements are achieved in comparison to baseline approaches.
\end{abstract}

\smallExtraSpacing

 

\maketitle

\section{Introduction}
\label{sec.intro}

Graphs are ubiquitous in modern applications due to their power in representing arbitrary relationships among entities. Organizing friendship relationships in social networks, modeling the Web, monitoring interactions among proteins are only a few application examples where the graph is a \textit{first class object}. This significant interest in graphs is the main motivation for the recent development of efficient algorithms for graph management and mining~\cite{CH06,AW10}.

A \textit{graph} or \textit{network} $G(V,E)$ is composed of a set of vertices $V$ and a set of edges denoted as $E$. In its simplest form, $G$ is \textit{undirected} (no direction is assigned to the edges) and \textit{unweighted} (the weight of each edge is assumed to be 1). Vertices represent entities whereas edges represent specific types of relationships between vertices. Due to their rich structural content, graphs provide significant opportunities for many important data mining tasks such as clustering, classification, community detection, frequent pattern mining, link prediction, centrality analysis and many more. 

Conventional graphs are characterized by the fact that both the set of vertices $V$ and the set of edges $E$ are known in advance, and are organized in such a way to enable efficient execution of basic tasks. Usually, the adjacency lists representation is being used, which is a good compromise between space requirements and computational efficiency. However, a concept that recently has started to gain significant interest is that of \textit{hidden graphs}. In contrast to conventional graphs, a hidden graph is defined as $G(V,f())$, where $V$ is the set of vertices and $f()$ is a function $V \times V$ $\rightarrow$ $\{0,1\}$ which takes as an input two vertex identifiers and returns true or false if the edge exists or not respectively. Therefore, in a hidden graph the edge set $E$ is not given explicitly and it is inferred by using the function $f()$. 

The brute-force approach for executing graph-oriented algorithmic techniques on hidden graphs comprises the following phases: 
\begin{enumerate}
\item[$i$)]
in the first phase, all possible $n(n-1)/2$ edge probes are executed in order to reveal the structure of the hidden graph completely, and 
\item[$ii$)]
in the second phase, the algorithm of interest is applied to the exposed graph. 
\end{enumerate}
It is evident, that such an approach is not an option, since the function $f()$ associated with edge probing queries may be extremely costly to evaluate and it may require the invocation of computationally intensive algorithms. The following cases are a few examples of hidden graph usage in real-world applications:

\begin{figure*}[!t]
\begin{center}
\begin{minipage}{5.6cm}
\centerline{\epsfig{file=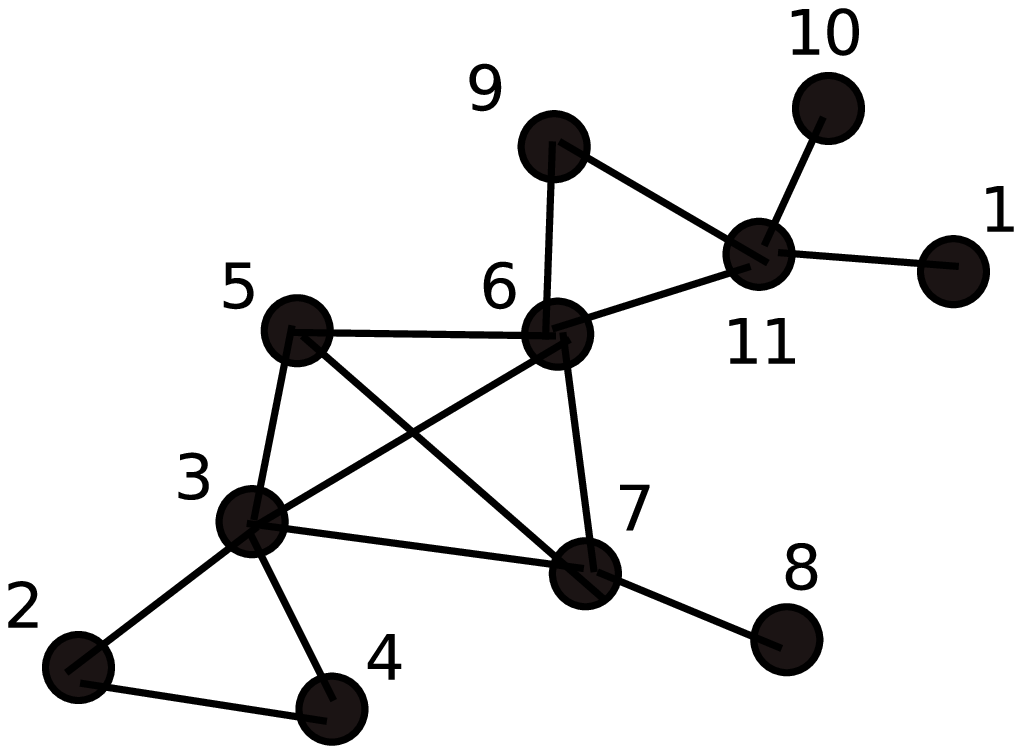,width=4.8cm}}
\centerline{}
\centerline{(a) 1-core}
\end{minipage}
\begin{minipage}{5.6cm}
\centerline{\epsfig{file=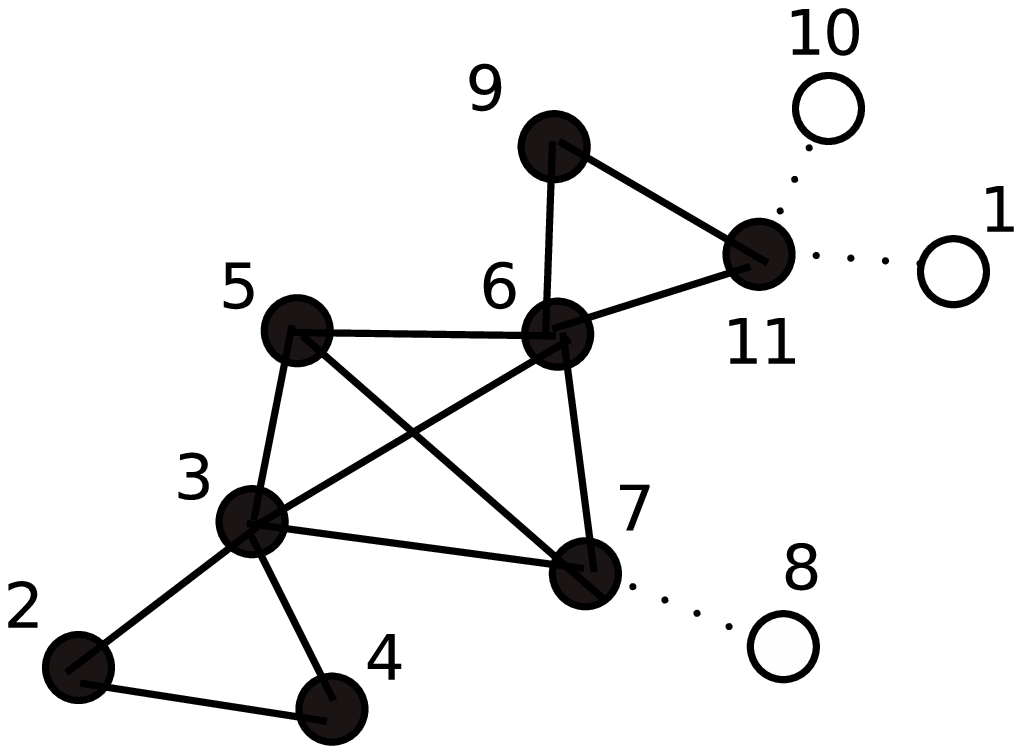,width=4.8cm}}
\centerline{}
\centerline{(b) 2-core}
\end{minipage} 
%\centerline{}
%\centerline{}
\begin{minipage}{5.6cm}
\centerline{\epsfig{file=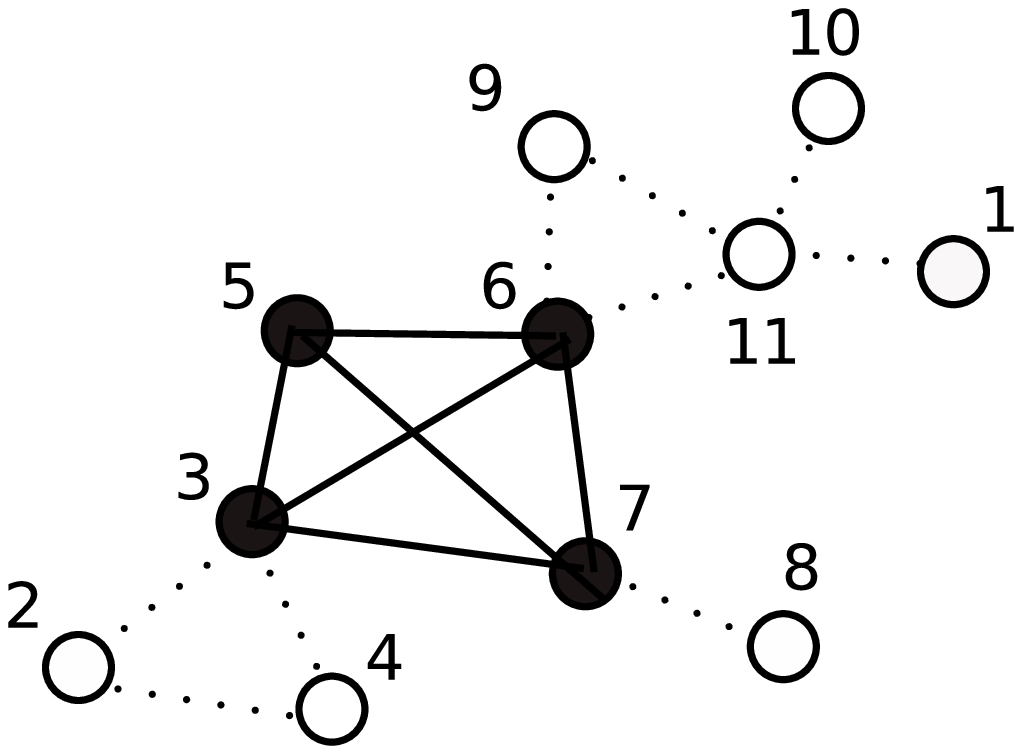,width=4.8cm}}
\centerline{}
\centerline{(c) 3-core (max core)}
\end{minipage}
\end{center}
\caption{Core decomposition example. Filled vertices are contained in the corresponding $k$-core. The maximum core of the graph is the 3-core and it is composed of vertices 3, 5, 6, and 7.}
\label{fig.examplecore}
\end{figure*}

\begin{itemize}
\item
A hidden graph may be defined on a document collection, where $f()$ returns 1 if the similarity between two documents is higher than a user-defined threshold and 0 otherwise. Taking into account that there are many diverse similarity measures that can be used, it is more flexible to represent the document collection as a hidden graph $G(V,f())$, using $f()$ to determine the level of similarity between documents.  
\item
Assume that the vertices of the graph are proteins, forming a protein-protein interaction network (PPI). An edge between two proteins $u$ and $v$ denotes that these proteins interact. Edge probing in this case corresponds to performing a lab experiment in order to validate if these two proteins interact or not. In this case, the computation of the function $f()$ is extremely costly.
\item
Relational databases may be seen as hidden graphs as well. In this scenario, the vertices of the hidden graph may be database records or entities, whereas the edges may correspond to arbitrary associations between these entities (corresponding to a $\theta$-join). For example, in an product-supplier database, vertices may represent suppliers and an edge between two suppliers may denote the fact that they supply at least a common product with specific properties. In this case, the function $f()$ involves the execution of a possibly complex SQL join query.
\item
As another example, consider the set of vertices defined by user profiles in a social network. A significant operation in such a network is the discovery of meaningful \textit{communities}. However, there is a plethora of methods to quantify the strength or similarity among users, ranging from simple metrics such as number of common interests to more complex ones like the similarities in their posts, or their mutual contribution in answering questions (like in the case of the Stack Overflow network). In these cases, taking into account that user similarity can be expressed in many diverse ways, the hidden graph concept is very attractive for ad-hoc community detection.  
\end{itemize}

Hidden graphs constitute an interesting tool and an promising alternative to conventional graphs, since there is no need to represent the edges explicitly. This enables the analysis of different graph types that are implicitly produced by changing the function $f()$. Note that the total number of possible graphs that can be produced for the same set of vertices equals $2^{\binom{n}{2}}$, where $n=|V|$ is the number of vertices. It is evident, that the materialization of all possible graphs is not an option, especially when $n$ is large. Therefore, hidden graphs is a tempting alternative to model relationships
among a set of entities. On the other hand, there are significant challenges posed, since the existence of an edge must be verified by evaluating the function $f()$, which is costly in general.

A significant graph mining task, which is highly related to community detection and dense subgraph discovery, is the \textit{core decomposition} of a graph. The output of the core decomposition process is a set of nested induced subgraphs (known as \textit{cores}) of the original graph that are characterized by specific constraints on the degree of the vertices. In particular, the 1-core of $G$ is the maximal induced subgraph $S_1$, where the degree of every vertex in $S_1$ is at least 1. The 2-core of $G$ is the maximal induced subgraph $S_2$, where all vertices have degree at least 2. In general, the $k$-core of $G$ is the maximal induced subgraph $S_k$ where the degree of every vertex in $S_k$ is at least $k$. In addition, for any two core subgraphs $S_i$ and $S_j$, if $i < j$ then $S_j \subset S_i$. 

Based on the fact that the cores of a graph are nested, the \textit{core number} of a node $u$ is defined as the maximum value of $k$ such that $u$ participates in the $k$-core. The maximum core value that can exist in the graph is also known as the \textit{graph degeneracy}~\cite{Farach-Colton2014}. Formally, the degeneracy $\delta^*(G)$ of a graph $G$ is defined as the maximum $k$ for which $G$ contains a non-empty $k$-core subgraph.

A core decomposition example is illustrated in Figure~\ref{fig.examplecore}. Black-colored nodes participate in the corresponding core. Therefore, the graph in Figure~\ref{fig.examplecore}(a) corresponds to the 1-core of $G$, since the degree of all nodes is at least one and there is no supergraph with this property. Similarly, Figures~\ref{fig.examplecore}(b) and \ref{fig.examplecore}(c) show the 2-core and the 3-core of $G$ respectively. Note that the induced subgraph representing the $k$-core contains nodes with degree at least $k$ and there is no larger subgraph with this property. Note also, that the maximum core of the graph is the 3-core, since beyond that point it is not possible to form an induced subgraph such that the degree of the nodes is at least 4. If one of the nodes participating in the maximum core is removed, the graph collapses. Thus, based on the definition of the graph degeneracy, in this case $\delta^*(G)=3$. 

\spara{Motivation and Contributions.}
The core decomposition of graphs has many important applications in diverse fields \cite{MPV16}. It has been used as an algorithm for community detection and dense subgraph discovery~\cite{Fortunato10}, as a visualization technique for large networks~\cite{vespignani-kcore-nips06}, as a technique to improve effectiveness in information retrieval tasks~\cite{Rousseau2015}, as a method to quantify the importance of nodes in the network in detecting influential spreaders~\cite{malliaros-www15}, and as a tool to analyze protein-protein interaction networks~\cite{ppi03}, just to name a few. 
In this work, we apply the concept of core decomposition in a hidden graph.
The fact that the graph $G$ is hidden, poses significant difficulties in the discovery of the $k$-core. First of all, since the edges are not known in advance, edge probing queries must be executed to reveal the graph structure. In addition, specialized bookkeeping is required in order not to probe an edge multiple times. Formally, the problem we attack has as follows: ~\\

\spara{PROBLEM DEFINITION}. \textit{Given a hidden graph $G(V,f())$, where $V$ is the set of nodes and $f()$ a probing function, and an integer $k$, discover a $k$-core of $G$ if such a core does exist, by using as few edge probing queries as possible.} ~\\

To the best of the authors' knowledge, this is the first work studying the problem of $k$-core computation in a hidden graph. In particular, we present the first algorithm to compute the $k$-core of a hidden graph, if such a core exists. 
Our solution is based on the following methodology: 

\begin{enumerate}
\item
Firstly, we generalize the SOE (switch-on-empty) algorithm proposed in~\cite{TSL10} in order to be able to determine nodes with high degrees in any graph, since the initial SOE algorithm supports only bipartite graphs and determines the largest degrees focusing only in one of the two bipartitions.
In addition, the generalized algorithm can also be applied in directed graphs with minor modifications.
\item
Secondly, we enhance the generalized algorithm (GSOE) with additional data structures in order to provide efficient bookkeeping during edge probing queries. 
\item
Finally, we provide the \textsc{HiddenCore} algorithm which takes as input an integer number $\mathcal{K}$ and either returns the $\mathcal{K}$-core of $G$ or \textit{false} if the $\mathcal{K}$-core does not exist. Although \textsc{HiddenCore} is based on GSOE, it uses different termination conditions and performs additional bookkeeping to deliver the correct result.
\end{enumerate}

Performance evaluation results have shown that significant performance improvement is achieved in comparison to the baseline approach which performs all possible $\mathcal{O}(n^2)$ edge probing queries in order to reveal the structure of the graph completely. ~\\

\spara{Roadmap.}
The rest of the paper is organized as follows. Section~\ref{sec.related} presents related work in the area, covering the topics of core decomposition and hidden graphs. Section~\ref{sec.basic} contains some background material useful for the upcoming sections. The proposed methodology is given in detail in Section~\ref{sec.proposed}. Performance evaluation results are offered in Section~\ref{sec.performance} and finally, Section~\ref{sec.conclusions} concludes the work and discusses briefly interesting topics for future research in the area.

\section{Related Work}
\label{sec.related}

Hidden graphs have attracted a significant attention recently, since they allow the execution of graph processing tasks, without the need to now the complete graph structure. This concept 
was originally introduced in~\cite{GK98}, where edge probing queries were used to test if there is an edge between two nodes. 

One research direction which uses the concept of edge probes is \textit{graph property testing} \cite{GGR98}, where one is interested to know if a graph $G$ has a specific property, e.g., if the graph is bipartite, if it contains a clique, if it is connected, and many more. However, in order to test if the graph satisfies a property or not, the number of edge probing queries must be minimized, leading to \textit{sublinear} complexity with respect to the number of probes. Moreover, these algorithms are usually probabilistic in nature and provide some kind of probabilistic guarantees for their answer, by avoiding the execution of a quadratic number of probes.  

Another research direction related to hidden graphs, focuses on \textit{learning} a graph or a subgraph by using edge probing queries using pairs or sets of nodes (group testing) \cite{AV05}. A similar topic is the \textit{reconstruction} of subgraphs that satisfy certain structural properties \cite{BGK05}. 

One of the problems related to reconstruction, is the discovery of the $k$ nodes with the highest degree. In \cite{TSL10}, the SOE (Switch-on-Empty) algorithm is proposed to solve this problem in a bipartite graph. It has been shown that SOE is significantly more efficient than the baseline approach which simply reveals the graph structure by executing all possible $\mathcal{O}(n^2)$ edge probing queries. The same problem has been also studied in \cite{YLW13} using combinatorial group testing, which allows edge probing among a specific set of nodes instead of just one pair of nodes.

The core decomposition is a widely used graph mining task with a significant number of applications in diverse fields~\cite{MPV16}. The concept was first introduced in~\cite{seidman-1983} and later on it was adopted as an efficient graph analysis and visualization tool~\cite{ADBV05,ZP12}.
The baseline algorithm to compute the core decomposition requires $\mathcal{O}(m \log n)$ operations and it is based on a minheap data structures with the decrease key operation enabled ($m$ is the number of edges and $n$ the number of node). The algorithm gradually removes the node with the smallest degree, updating node degrees as necessary.  
A more efficient algorithm with linear $\mathcal{O}(n+m)$ complexity was proposed in \cite{BZ03}. The algorithm uses bucket sorting and multiple arrays of size $n$ to achieve linearity. 

There is a plethora of algorithms for the computation of the core decomposition under different settings and computational models. Some of these efforts are: disk-based computation~\cite{CKCO11}, incremental core decomposition~\cite{SGJW+13}, distributed core decomposition computation~\cite{MPM13,PKT14}, local core number computation~\cite{OS14},
core decomposition of uncertain graphs~\cite{BGKV14}.

The main characteristic of the aforementioned core decomposition algorithms is that in order to operate, the set of edges must be known in advance. 
In the sequel, we present our solution for detecting cores in hidden graphs which is based on edge probing queries and it does not requires knowledge of the complete set of edges.

\section{Fundamental Concepts}
\label{sec.basic}

In this section, we discuss some fundamental concepts necessary for the the upcoming material.
In particular, we will present briefly the use of the Switch-On-Empty algorithm proposed in~\cite{TSL10} and also we will discuss the linear core decomposition algorithm reported in~\cite{BZ03}.

\subsection{Preliminaries}

\begin{table}[!b]
\begin{center}
\caption{Frequently used symbols.}
\label{tab.symbols}
\renewcommand{\arraystretch}{1.2}
\begin{tabular}{|c||l|}
\hline
{\bf Symbol}	& {\bf Interpretation}                             					\\ \hline\hline
$G$				& a hidden graph													\\ \hline
$V$				& set of vertices of $G$												\\ \hline
$n$				& number of vertices of $G$ ($n = |V|$)								\\ \hline
$u,v$			& vertices of $G$													\\ \hline
$N(u)$			& set of neighbors of vertex $u$									\\ \hline
$d(u)$			& degree of vertex $u$  												\\ \hline
$E$				& (unknown) set of edges 											\\ \hline
$m$				& (unknown) number of edges ($m = |E|$)								\\ \hline
$v_s$, $v_d$    & source and destination vertices									\\ \hline
$f(v_s,v_d)$	& {\it true} if the edge $(v_s,v_d)$ exists, {\it false} otherwise	\\ \hline
$k$				& number of highest degree vertices requested 						\\ \hline
$\mathcal{K}$	& defines the $\mathcal{K}$-core of $G$								\\ \hline
$s(u)$			& number of known existing neighbors of $u$							\\ \hline
$e(u)$			& number of known non-existing neighbors of $u$						\\ \hline
$probes$		& total number of edge probing queries issued						\\ \hline
\end{tabular}
\end{center}
\end{table}

The input hidden graph is denoted as $G$, and contains $n$ vertices and $m$ edges. The number
of neighbors of $u$ is known as the degree of $u$, $d(u)$. Note that some quantities are not known in advance. For example, the total number of edges $m$, vertex degrees, the diameter, and any 
value related to the graph edges is unknown. Table~\ref{tab.symbols} summarizes the most frequently used symbols.

Initially, the number of neighbors of each vertex is unknown. As edge probing queries are executed, the graph structure gradually reveals. When an edge probing query between vertices $u$ and $v$ is executed (i.e., the function $f(u,v)$ is invoked), either the edge $(u,v)$ exists or not. 
Two counters are associated with each vertex $u$: the counter $s(u)$ counts the number of edges that are \textit{solid}, i.e., they exist and the counter $e(u)$ counts the number of \textit{empty},
i.e., non-existing edges incident to $u$. Therefore, if $(u,v)$ exists, then the counters 
$s(u)$ and $s(v)$ are incremented. Otherwise the counters $e(u)$ and $e(v)$ are incremented.
The sum $s(u)+e(u)$ measures the number of edge probing queries executed where $u$ is one endpoint.

\subsection{The Switch-On-Empty Algorithm (SOE)}

Before diving into the details of the GSOE algorithm, the original SOE algorithm, proposed in \cite{TSL10}, is described briefly. In SOE, the input is a bipartite graph, with bipartitions $A$ and $B$. The output of SOE is composed of the $k$ vertices from $A$ or $B$ with the highest degree. Without loss of generality, assume that we are focusing on vertices in $A$. Edge probing queries are executed as follows: 

\begin{itemize}
\item
SOE starts from a vertex $a_1 \in A$, selects a vertex $b_1 \in B$ and executes $f(a_1, b_1)$. If the edge $(a_1,b_1)$ is solid, it continues to perform probes between $a_1$ and another vertex $b_2 \in B$. 
\item
Upon a failure, i.e., when the probe $f(a_1, b_j)$ returns an empty result, the algorithm applies the same for another vertex $a_2 \in A$. Vertices for which all the probes have been applied, do not participate in future edge probes. 
\item
A round is complete when all vertices of $A$ have been considered. After each round,
some vertices can be safely included in the result set $R$ and they are removed from $A$. 
When a vertex $a_1$ must be considered again, we continue the execution of probes remembering the location of the last failure. 
\item
SOE keeps on performing rounds until the upper bound of vertex degrees in $A$ is less than the current $k$-th highest degree determined so far. In that case, $R$ contains the required answer and the algorithm terminates.
\end{itemize}

The basic idea behind SOE, is that as long as probes related to a vertex are successful, we must continue probing using that vertex since there are good chances that this is a high-degree vertex.
It has been proven in \cite{TSL10}, that SOE is \textit{instance optimal}, which means that on any hidden bipartite graph given as an input, the algorithm is as efficient as the optimal solution, up to a constant factor. It has been shown, that this constant is at most two for any value of the parameter $k$ (number of vertices with the highest degree). 

\begin{figure}[!ht]
\begin{center}
\includegraphics[scale=0.5]{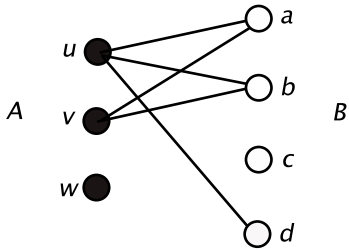}
\end{center}
\caption{An example of an undirected bipartite graph, with two bipartitions $A=\{u,v,w\}$ and $B=\{a,b,c,d\}$, where $d(u)=3$, $d(v)=2$ and $d(w)=0$. Vertex $u$ is the one with the highest degree.}
\label{fig.soe}
\end{figure}

In the sequel, we provide a simple example to demonstrate the way SOE works to discover the top-$k$ vertices with the highest degree. Let $G$ denote a hidden bipartite graph, containing $n=7$ vertices and $m=5$ edges as shown in Figure~\ref{fig.soe}. We assume that in our case $k=1$, i.e., we need to detect the vertex with the highest degree. Without loss of generality we focus on the left bipartition (vertex set $A$) which contains the vertices $u$, $v$ and $w$.

If we apply the brute-force algorithm in this graph, we need to perform all $3 \times 4 = 12$ edge probes first, and then simply select the vertex with the highest degree among the subset $\{u,v,w\}$. In contrast, SOE will perform the following sequence of probes: $f(u,a)=solid$, $f(u,b)=solid$, $f(u,c)=empty$, $f(v,a)=solid$, $f(v,b)=solid$, $f(v,c)=empty$, $f(w,a)=empty$,
$f(u,d)=solid$. At this stage, SOE knows that vertex $u$ will be part of the answer since the degree of $u$ can be computed exactly, since all probes related to vertex $u$ have been executed. The next probe will be $f(v,d)=empty$ and know SOE can eliminate vertex $v$ since its degree cannot be larger than 3 which is the degree of $u$. The next probe is $f(w,b)=empty$, and now SOE terminates since vertex $w$ cannot make it to the answer since $d(w) < d(u)$. The total number of probes performed by SOE is 10, whereas the brute-force algorithm requires 12.

\subsection{Cores in Conventional Graphs}

The core decomposition of a conventional graph can be computed in linear time, as it is discussed thoroughly in~\cite{BZ03}. The pseudocode is given in Algorithm~\ref{algo.batagelj} (\textsc{CoreDecomposition}). To achieve the linear time complexity, comparison-based sorting is avoided and instead binsort is applied for better performance, since the degree of every vertex lies in the interval $[1,n-1]$ (isolated vertices are not of interest), where $n$ is the number of vertices.

\begin{algorithm}[!ht]
\caption{\textsc{CoreDecomposition}~($G$)}
\label{algo.batagelj}
\do
\dontprintsemicolon

\KwIn{the graph $G$}
\KwResult{the core numbers (array $C$)}

$V$ $\leftarrow$ set of vertices of $G$ \;
array $D$ $\leftarrow$ vertex degrees \;
sort array $D$ in non-decreasing order \;   
\For {each $v \in V$ in the order}
{
     $C[v]$ $\leftarrow$ $D[v]$ \;
     \For {each $u \in N(v)$}
     {
          \If {$D[u] > D[v]$}
          {
               $D[u] \leftarrow D[u] - 1$ \;
               reorder array $D$ accordingly \;
          }
     }
}
\Return $C$ \;
\end{algorithm}

Each time, the vertex with the smallest degree is selected and removed from the graph. The selection of the next vertex to remove, is performed in $\mathcal{O}(1)$. After vertex removal, the degrees of neighboring vertices are adjusted properly and for each neighbor a reordering is performed, again in $\mathcal{O}(1)$ time, due to the  usage of the bins. Each bin contains vertices with the same degree. Thus, there are at most $n-1$ bins. Since each edge is processed exactly once, the overall time complexity of the decomposition process is $O(n+m)$ for a graph containing $n$ vertices and $m$ edges.

The linear complexity combined with the usefulness of the decomposition process results in a very efficient process. However, in our case this technique can be applied only when the set of edges is known to the algorithm. In the next section, we present our proposal towards detecting $k$-cores in a general hidden graph.

\section{Proposed Methodology}
\label{sec.proposed}

In this section, we present our methodology in detail. Firstly, we focus on the generalization of the SOE algorithm. The \textit{Generalized Switch-On-Empty Algorithm} (GSOE) is able to find the top-$k$ degree vertices in an undirected hidden graph, whereas SOE can be applied on bipartite graphs only. Secondly, we present the algorithmic techniques to enable the discovery of vertices belonging to the $k$-core of the graph, if the $k$-core does exist. 

\subsection{Bookkeeping}
Since SOE focuses only on one of the two bipartitions of the input graph, the bookkeeping process is very simple, because it just needs to remember the last failure of every vertex. However, in a general graph $G$, this cannot be applied, because edge probes may affect the neighborhood list of other vertices. The aim of GSOE is to discover the $k$ vertices with the highest degree among all graph vertices, by performing as few edge probing queries as possible and by avoiding probing the same link twice. Each edge probing query is performed from a source vertex $v_s$ towards a destination vertex $v_d$, by invoking the function $f(v_s,v_d)$. Similarly to SOE, if the probe indicates that there is a connecting edge between $v_s$ and $v_d$, then this edge is considered as \textit{solid} otherwise it is marked as \textit{empty}. Based on the probing result, the algorithm either continues with the same source vertex and a different destination vertex or changes the source vertex as well and selects the next available one. 

For the proper selection of source and destination vertices, GSOE stores probing results at vertex-level data structures. As the algorithm evolves, these data structures store the necessary information required for the next selection of source and destination vertices. The result of GSOE is a set $R$ containing the $k$ vertices with the highest degrees, sorted in non-increasing order. To provide the final result, GSOE maintains the following information for every vertex $u$: 

\begin{itemize}
\item
the counter $s(u)$, monitoring the total number of solid edges incident to $u$,
\item
the set of solid edges, $SE(u)$, detected so far for vertex $u$ ($s(u)=|SE(u)|$),
\item
the counter $e(u)$, counting the total number of empty edges incident to $u$,
\item
the variable $state(u)$ is decreased by one whenever $u$ participates in an edge probe for which the edge does not exist, 
\item
an auxiliary data structure $PMS(u)$ (\textit{probe monitoring structure}) to be able to detect the next available vertex to act as destination, in order to perform the next edge probing query. 
\end{itemize}

For vertex $u$, the structure $PMS(u)$ performs the necessary bookkeeping regarding the probes performed so far related to $u$. Whenever $u$ participates in a probe either as source or destination vertex, $PMS(u)$ is updated accordingly. For the rest of the discussion, we will assume that vertex identifiers take values in the interval $[1,n]$, where $n=|V|$ is the total number of vertices of the hidden graph. Let $EE(u)$ denote the set of empty edges detected for $u$. Note that, we use this set for the convenience of the presentation, since it is not being used by the algorithm. 

First, we focus on the selection of a destination vertex, assuming that the source vertex is known. Later, we will also discuss thoroughly how source vertices are selected. Let $u$ be the selected source vertex. We are interested in determining a vertex $v$ in order to issue the probing query $f(u,v)$. 

The next destination vertex $v$ must satisfy the following property: $v \notin SE(u) \cup EE(u)$, i.e., $v$ must not have been considered previously. 
The straight-forward solution to detect $v$, is to consider the union $SE(u) \cup EE(u)$ and find the first available vertex identifier. This solution has a time complexity of $\mathcal{O}(s(u)+e(u))$, because both sets $SE(u)$ and $EE(u)$ must be scanned once. Taking into account that $s(u)+e(u)$ can be as large as $n-1$, we are interested in a more efficient design. 

Assume that $u$ is part of a hidden graph with $n=10$ vertices. For the purpose of the example, let $u=1$. Assume further, that at a specific instance the status of the probes is $SE(u)=\{2,8\}$ and $EE(u)=\{6\}$, which means that $s(u)=2$ and $e(u)=1$. Since $n=10$, there are still six available vertices to be selected as destinations.  Thus, the set of solid and empty edges define a set of \textit{available intervals} containing vertex identifiers that can be selected as destinations. Based on our example, the set of available intervals has as follows: $A = \{[3,5]$, $[7,7]$, $[9,10]$\}. Since intervals are pair-wise disjoint and they never overlap, they can be organized in a balanced binary search tree data structure, where the key corresponds to the left (right) endpoint regarding the left (right) subtrees.  

\begin{figure}[!t]
\begin{center}
\begin{minipage}{4cm}
\centerline{\epsfig{file=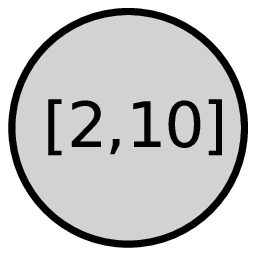,width=1cm}}
\centerline{(a) selected: 2}
\end{minipage}
\begin{minipage}{4cm}
\centerline{\epsfig{file=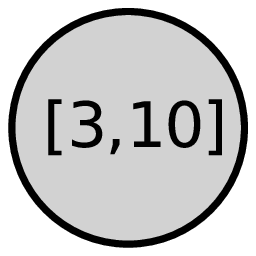,width=1cm}}
\centerline{(b) selected: 5}
\end{minipage}
\centerline{}
\centerline{}
\begin{minipage}{4cm}
\centerline{\epsfig{file=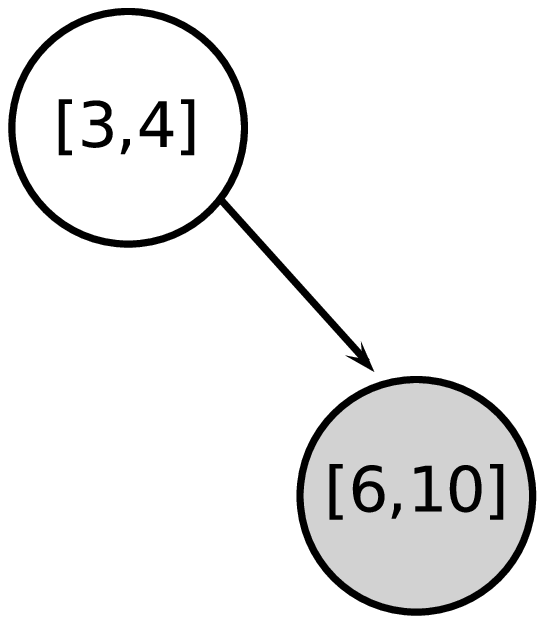,width=2.2cm}}
\centerline{(c) selected: 8}
\end{minipage}
\begin{minipage}{4cm}
\centerline{\epsfig{file=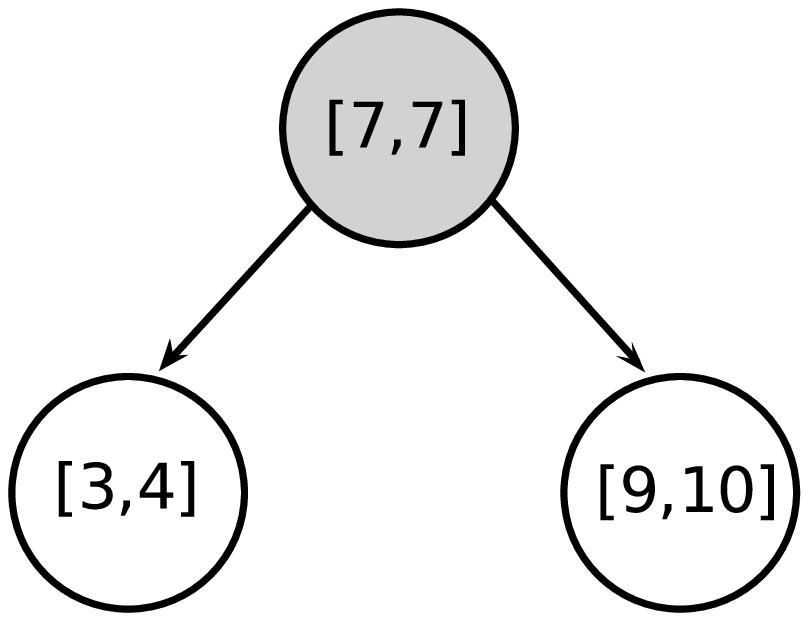,width=3cm}}
\centerline{(d) selected: 7}
\end{minipage}
\centerline{}
\centerline{}
\begin{minipage}{4cm}
\centerline{\epsfig{file=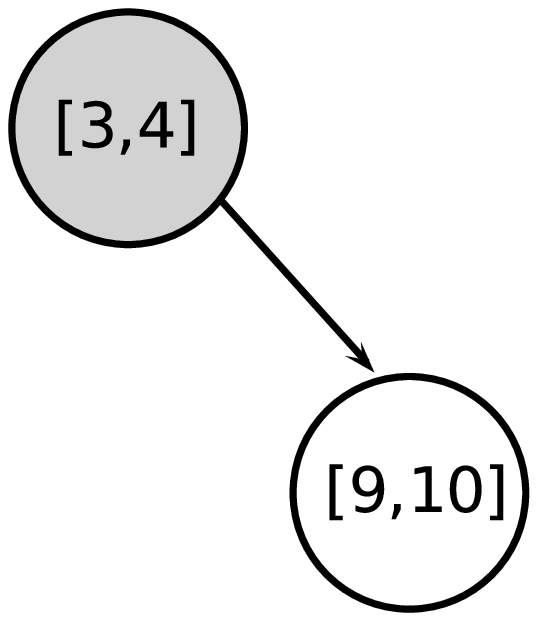,width=2.2cm}}
\centerline{(e) selected: 3}
\end{minipage}
\begin{minipage}{4cm}
\centerline{\epsfig{file=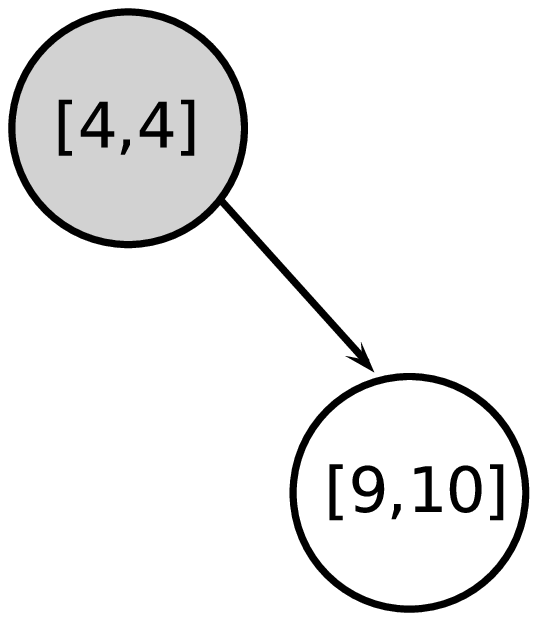,width=2.2cm}}
\centerline{(f) selected: 4}
\end{minipage}
\end{center}
\caption{A sequence of destination selections. Each time, the selected destination vertex is removed from the BST. The BST node containing the selected vertex is shown gray.}
\label{fig.intervals}
\end{figure}
 
An example is illustrated in Figure~\ref{fig.intervals}, showing a sequence of destination selections. Initially, the set of available intervals contains only the interval $[2,10]$, whereas $SE(u) = \emptyset$ and $EE(u) = \emptyset$. Let $v=2$ be the first vertex selected as destination and that the probe $f(u,v)$ returns a solid result. This means that 
$SE(u) = \{2\}$ and $EE(u) = \emptyset$. In fact, for the maintenance of the BST, it does not matter of we have solid or empty edges. All that matters is the vertex being selected as a destination. In this example, the selection order of destinations is arbitrary and the selected vertices are: 2, 5, 8, 7, 3 and 4. In Figure~\ref{fig.intervals}, we observe the evolution of the BST as destination vertices are being deleted gradually from the set of available intervals. 

\begin{lemma}
\label{lemma1}
Given a source vertex $u$, the selection of the destination vertex requires $\mathcal{O}(1)$ time, whereas updating the information of the available destinations requires $\mathcal{O}(\log (s(u)+e(u)))$ time.
\end{lemma}
\begin{proof}
Since each node of the BST contains an interval of available destinations, it suffices to visit the root and select a destination from the corresponding interval of the root node. Evidently, this operation takes constant time. We distinguish between two cases: $i$) the interval is of the form $[x,y]$ where
strictly $x < y$ and $ii$) the interval is of the form $[x,x]$. In the first case, we select as destination either $x$ or $y$ in order to avoid any structural operations on the BST. Thus, the length of the interval is reduced by one. In the second case, the interval $[x,x]$ is deleted from the BST. The number of elements in the BST is at most $s(u)+e(u)+1$, which means that deletions require $\mathcal{O}(\log (s(u)+e(u)))$ time in the worst case. 
\end{proof}

After updating $PMS(u)$ for the destination vertex $v$, the edge probing query $f(u,v)$ is issued. If $f(u,v) = true$, $v$ is inserted into $SE(u)$ and also, $u$ is inserted into $SE(v)$. In addition, the $PMS(v)$ must be updated as well, which means that the BST associated with vertex $v$ must exclude vertex $u$ from the available destinations. To facilitate this operation, a lookup in the BST is performed for the key $u$, in order to detect the interval containing $u$. Note that, since intervals are disjoint, $u$ is contained in one and only one interval, which can be detected in logarithmic time $\mathcal{O}(\log (s(v)+e(v)))$.  

We distinguish among three different cases: $i$) $v$ is included in the interval $[v,x]$ or $[x,v]$ and in this case the interval is simply shrunk from the right or the left endpoint respectively. $ii$) $v$ is included in the interval $[v,v]$, and a single deletion of the vertex is required. $iii$) $v$ is included in the interval $[x,y]$ and $v \ne x$ and $v \ne y$. In this case, the interval $[x,y]$ is split to two intervals $[x,v-1]$ and $[v+1,y]$. The original interval $[x,y]$ is deleted from the BST whereas the two new subintervals are inserted in the BST. In any case, the cost is $\mathcal{O}(\log (s(v)+e(v)))$

\begin{theorem}
Given a source vertex $u$, the selection of the destination vertex and the updates of the structures $PMS(u)$ and $PMS(v)$ take time $\mathcal{O}(\log n)$, where $n$ is the number of graph vertices.
\end{theorem}
\begin{proof}
The result follows from Lemma~\ref{lemma1} and from the fact that the number of intervals that can be hosted by each BST is at most $\frac{n}{2}$.
\end{proof}

So far, we have focused on the selection of a destination vertex, assuming that the source vertex is already known. Next, we elaborate on the selection of the source vertex to participate in the next edge probing query.

\subsection{Detecting High-Degree Vertices}

Let $R$ denote the result set containing at least $k$ vertices with the highest degrees.
Note that, in case of ties (i.e., if many vertices have the same degree as the $k$-th),
these vertices will be also included in $R$. Let $u$ be the current source vertex. As long as the probing queries $f(u,v)$ return solid edges, the source vertex remains the same and the structures $SE(u)$, $PMS(u)$, $SE(v)$ and $PMS(v)$ are updated accordingly, as described in the previous section. 

If $f(u,v) = false$, i.e., the edge $(u,v)$ does not exist, the source vertex should change and another vertex is selected as source. In GSOE, the following rules are applied: ~\\

\begin{enumerate}
\item[\textit{Rule 1}] When the probing $f(u,v)$ comes out solid, the values $s(u)$ and $s(v)$ are increased by one.
\item[\textit{Rule 2}] When the probing $f(u,v)$ comes out empty, the values $e(u)$, $e(v)$,
$state(u)$ and $state(v)$ are decreased by one.
\item[\textit{Rule 3}] A vertex $u$ can be pushed to the result set $R$ if GSOE found its actual degree and $state(u) = 0$.
\item[\textit{Rule 4}] When a vertex $u$ has $state(u) < 0$ it cannot be selected as a source vertex. It can only be selected as a destination vertex. A vertex $u$ can be selected as a source vertex only if $state(u)=0$ and it does not fulfill \textit{Rule 3}.
\end{enumerate}

GSOE terminates when no more vertices can be added in the result set $R$. This means that the maximum potential degree for a vertex $v \notin R$ is strictly less than the $k$-th best degree contained in $R$.  

The sequence of probing queries is performed in \textit{rounds}. If $u$ is the first vertex to be checked as a potential sourse, a round is complete when $u$ is checked again as a potential source vertex. If it fulfills the necessary requirements stated by \textit{Rule 4} above, then it will be selected as the next source.  

\begin{lemma}
\label{lemma2}
For every vertex $u \notin R$, it holds that $e(u)+1$ equals the number of rounds spent 
while $u \notin R$.
\end{lemma}
\begin{proof}
From the definition of $state(u)$, we conclude that when $e(u)$ is increased by one then $state(u)$ is decreased by one. This means that $state(u)$ is decreased $e(u)$ times. According to \textit{Rule 3}, GSOE pushes a vertex to $R$ only if this vertex has zero state. Recall that a negative $state$ value is increased by one at the end of each round. Thus, for a vertex $u$ that is pushed in $R$ it holds that the value $e(u)+1$ equals
the number or rounds performed with $u \notin R$. In case $u \in R$ after GSOE terminates, then $e(u)+1$ equals the number of rounds needed for $u$ to be included in $R$.
\end{proof}

\begin{lemma}
\label{lemma3}
Two or more vertices are pushed in $R$ during the same round if and only if they have the same degree.
\end{lemma}
\begin{proof}
We provide the proof for two vertices $u$ and $v$ since the generalization is obtained easily.
Assume that $u$ and$v$ are pushed to $R$ during the same round. This means that GSOE had spent the same number of rounds until it pushes them to $R$. Based on \ref{lemma2} we conclude that:
$e(u)+1 = e(v)+1 \implies e(u)=e(v)$. Since $u$ and $v$ are both contained in $R$ it holds that:
$s(u)+e(u)$=$n-1$ and $s(v)+e(v)$=$n-1$, which means that $s(u)+e(u)$ = $s(v)+e(v)$ and thus,
$s(u) = s(v)$. 
\end{proof}

\begin{algorithm}[!t]
\caption{\textsc{Update}($v_s,v_d,proberesult$)}
\label{algo.update}
\DontPrintSemicolon
\BlankLine
\KwIn{source $v_s$, destination $v_d$, probe result}
\KwResult{update the probe monitoring structures}
\BlankLine
\If {proberesult = solid} {
  $SE(v_s) \leftarrow SE(v_s) \cup \{v_d\}$ \;
  $SE(v_d) \leftarrow SE(v_d) \cup \{v_s\}$ \;
}
insert $v_d$ to $PMS(v_s)$ \;
insert $v_s$ to $PMS(v_d)$ \;
\end{algorithm}

The usefulness of the previous lemma lies in the fact that all vertices having the same degree as the $k$-th best vertex, will enter the result set during the same round, and therefore the termination condition of the algorithm is based only on the number of elements contained in set $R$. Consequently, the condition $|R| \geq k$ is sufficient to terminate and it guarantees that the result set $R$ is correct.    

\begin{algorithm}[!t]
\caption{GSOE($G(V,f)$, $k$)}
\label{algo.gsoe}
\DontPrintSemicolon
\BlankLine
\KwIn{the hidden graph $G$, the number $k$}
\KwResult{the set $R$ with highest degree vertices}
\BlankLine
$n \leftarrow |V|$ \;  
\While{true} {
    $v_s \leftarrow$ select a vertex $v_s$, where $state(v_s)=0$ \;  
    \While{($s(v_s)+e(v_s) = n-1$)} {
    	$R$ $\leftarrow$ $R \cup \{v_s\}$  \com{/* insert source to results */}\;
        %\Comment {see Rules 3 and 4}
        $v_s \leftarrow$ select another vertex with zero state \;
    }
    \While{$v_s \leqslant$ last vertex in $G$} { 
    	$v_d \leftarrow$ select destination vertex using $v_s$ \;
        $proberesult \leftarrow f(v_s,v_d)$ \;
        \textsc{Update}$(v_s, v_d, proberesult)$ \;
 		\If{$probresult$ is $solid$} {
    		$s(v_s)$++ \;
            $s(v_d)$++ \;
            %\Comment {see Rule 1}            
            \If{$s(v_s) + e(v_s) = n-1$} {
                %\Comment {see Rule 3}
            	$R$ $\leftarrow$ $R \cup \{v_s\}$ \com{/* insert source to results */} \;
            	} %\EndIf 
            \If{$s(v_d) + e(v_d) = n-1$ and $state(v_d) = 0$} {
                %\Comment {see Rule 3}
            	$R$ $\leftarrow$ $R \cup \{v_d\}$ \com{/* insert dest to results */} \;             
            } %\EndIf
           	\If{$v_s =$ last vertex in $G$} {
           		\If{$|R| \geq k$} {
            		\Return $R$ \;
            	} %\EndIf 
            } %\EndIf                           
            } %\EndIf          
  		\Else {
    		$e(v_s)++$ \;
            $e(v_d)++$ \;
            $state(v_s)--$ \; 
            $state(v_d)--$ \;
            %\Comment {see Rule 2}
            $v_s \leftarrow$ the next vertex with zero state \;
  		} %\EndIf 
    } %\EndWhile     
} %\EndWhile
\end{algorithm}

The outline of GSOE is given in Algorithm~\ref{algo.gsoe}. New vertices are inserted into the result set at Lines 5, 15 and 17. The termination condition is checked at Lines 19 and  if it is satisfied the algorithm returns the set $R$ containing the high-degree vertices, otherwise it continues with the next source vertex. After each probe, the bookkeeping structures are updated accordingly at Line 10 where the \textsc{Update} function (shown in Algorithm~\ref{algo.update}) is invoked.

\subsection{Core Discovery}

In the previous section, we discussed a solution for solving the problem of detecting the $k$ vertices with the highest degrees in a hidden graph $G$, using the Generalized Switch-On-Empty algorithm. In this section, we dive into the problem of discovering the $\mathcal{K}$-core of $G$, if such a core does exist. We remind that the $\mathcal{K}$-core of $G$ is the maximal induced subgraph $S$ where for each vertex $u \in S$,
$d(u) \geq \mathcal{K}$. To attack the problem, we propose the \textsc{HiddenCore} algorithm, which extends GSOE by using different criteria for selecting source and destination vertices and different termination conditions to guarantee efficiency and correctness.

In order for a vertex $u$ to belong to the $\mathcal{K}$-core, $d(u) \geq \mathcal{K}$. After an empty probe, \textsc{HiddenCore} estimates the maximum degree value that source and destination vertices could reach. For this purpose, \textsc{HiddenCore} introduces a new vertex-level parameter called \textit {potential degree}. In a hidden graph $G$ with $n$ vertices,  it holds that $\forall u \in G, pd(u)=n-1-e(u)$. If the potential degree of a vertex becomes less than $\mathcal{K}$, then \textsc{HiddenCore} blacklists this vertex. This practically means that probings from or towards this vertex is useless and thus, this vertex will be ignored for the rest of the algorithm execution. The value of $pd(u)$ is updated every time there is en empty probe related to vertex $u$, i.e., $u$ participates either as source or destination vertex.

Based on the definition of the $\mathcal{K}$-core, in order for a graph to have a $\mathcal{K}$-core there must be at least $\mathcal{K}+1$ vertices with degree greater than or equal to $\mathcal{K}$. Once \textsc{HiddenCore} realizes that it is impossible to satisfy this property, it terminates with a false result, since the $\mathcal{K}$-core does not exist in $G$. To enable this process, we introduce the concept of the \textit{number of potential core vertices}, symbolized as $\mathcal{C}$. Initially, 
$\mathcal{C} = n$, since all vertices are candidates to be included in the $\mathcal{K}$-core.
Gradually, as more empty probes are introduced, whenever for a vertex $u$, $pd(u) < \mathcal{K}$, the value of $\mathcal{C}$ is decreased by one. Consequently, if during the course of the algorithm the value of $\mathcal{C}$ becomes less than $\mathcal{K}+1$, the algorithm terminates with a false result, since it is impossible to detect the $\mathcal{K}$-core in $G$.

The aforementioned termination condition cannot restrict the number of probes issued as long as $\mathcal{C} \geq \mathcal{K}+1$. This means that \textsc{HiddenCore} will terminate when all possible $n(n-1)/2$ probes are executed. To handle this case, we introduce the \textit{maximum potential degree} variable which is the maximum value of $pd(u)$, for all vertices of $G$ not yet in $R$, and it is formally defined as follows: 
$$
mpd(G) = \max_{u \notin R} pd(u)
$$
The value of $mpd(G)$ is checked at the end of every round. If $mpd(G) < \mathcal{K}$, then we know that no additional vertices will ever satisfy the conditions to enter the result. Consequently, \textsc{HiddenCore} can proceed by executing the \textsc{CoreDedomposition} algorithm. The first condition is applied after every probe, whereas the second one is applied after each round. The combination of the two aforementioned termination conditions leads to a significant reduction in the number of probes.

Note that, \textsc{HiddenCore} is able to detect vertices that cannot make it to the final result by examining their potential to raise the number of solid edges to $\mathcal{K}$. However, a more aggressive termination condition can be applied that takes into account the potential of finding at least $\mathcal{K}+1$ vertices, based on the number of probes still available. More specifically, let $T$ denote the set of $\mathcal{K}+1$ vertices with the highest number of solid edges detected. These vertices can be effectively organized using a minheap data structure, which is updated after each probe. The vertices in $T$ define a lower bound on the number of probes required in order for the \textsc{CoreDecomposition} algorithm to be applied. Any vertex $u \in T$, requires $\mathcal{K} - s(u)$ additional probes to have chances to increase its degree above $\mathcal{K}$. Therefore, the total requirements with respect to the minimum number of additional probes needed is given by the following formula (note: $nrp$ stands for the \textit{number of required probes}):
$$
nrp = \frac{\sum_{\forall u \in T} (\mathcal{K} - s(u))}{2}
$$
This number must be less than or equal to the \textit{number of available probes} ($nap$), that we can still issue. Evidently, it holds that:
$$
nap = n(n-1)/2 - probes
$$
Based on the previous discussion, \textsc{HiddenCore} must terminate its execution whenever $nrp > nap$. The value of $nrp$ can be monitored efficiently by updating the contents of the set $T$, and this requires logarithmic time with respect to the size of $T$, which is $\mathcal{K}+1$.

During the course of the algorithm, a subgraph $S(V_S,E_S)$ is constructed, which accommodates all graph vertices $u$ where $d(v) \geq \mathcal{K}$. Also, $V_S$ satisfies the constraint $|V_S| \geq \mathcal{K}+1$. It is important to note that the subgraph $S$ does not contain any hidden edges, and therefore no additional probes are required to reveal its structure completely. The last phase of \textsc{HiddenCore} involves the execution of the \textsc{CoreDecomposition} algorithm (Algorithm~\ref{algo.batagelj}), in order to decide if the $\mathcal{K}$-core exists or not. This is necessary, since the degree constraint of vertices in $S$ involves the whole graph $G$ and not the subgraph induced by $V_S$. The following lemma guarantees the correctness of the result returned by \textsc{HiddenCore}. 

\begin{lemma}
The $\mathcal{K}$-core of $G$ exists, if at least $\mathcal{K}+1$ vertices in $V_S$ have a core number greater or equal to $\mathcal{K}$. 
\end{lemma}
\begin{proof}
In case all vertices in $S$ have a core number exactly $\mathcal{K}$, we are done since this is the definition of the $\mathcal{K}$-core. However, it may be the case that \textsc{CoreDecomposition} decides that all core numbers are strictly larger than $\mathcal{K}$. 
This means that higher-order cores are available, and due to the fact that cores are hierarchically nested, also lower-order cores must exist as well, and therefore no vertex will be missed.
\end{proof}

\begin{algorithm}[!t]
\caption{\textsc{HiddenCoreCheck}($u$, $\mathcal{K}$)}
\label{algo.hiddencorecheck}
\DontPrintSemicolon
\BlankLine
\KwIn{the hidden graph $G$, the number $\mathcal{K}$}
\KwResult{the set $R$ with highest degree vertices}
\BlankLine
$pd(u) \leftarrow n-1-e(u)$ \;
\If{$pd(u) \geq \mathcal{K}$} {
	\If{$pd(u) > mpd(G)$} {
		$MaxPotentialDegree \leftarrow pd(u)$ \;
        %\Comment {MaxPotentialDegree is initialized to 0 when a new round starts}
	}
    \Return $true$ \;
   }
\Else {
	$\mathcal{C}--$ \;
    %\Comment {PotentialCoreNodes equals to |V|=N at the start of GSOE}
	\If{$\mathcal{C} < K+1$} {
		\Return $false$ \;
		%\Comment {Criterion 2: Core does not exist}        
	}	
}
\end{algorithm}

A pleasant side effect of the above result is that if the $\mathcal{K}$-core does exist, due to the execution of the \textsc{CoreDecomposition} algorithm, higher order cores are also directly available. Therefore, \textsc{HiddenCore} is able to compute the complete core decomposition of the hidden graph for the subset of vertices that are contained in the $\mathcal{K}$-core of $G$. 

Also, we note that \textsc{HiddenCore} can be used for core discovery in \textit{hidden directed graphs} as well, where edge directionality is important. Evidently, the definition of the result and the termination conditions should be updated accordingly to reflect the fact that each vertex contains a set of outgoing edges, and a set of incoming edges. The concept of core decomposition in directed graphs has been covered in~\cite{GTV11} and it has many important applications, since a significant part of real-world graphs are directed.

\begin{algorithm}[!t]
\caption{\textsc{HiddenCore}$(G(V,f()), \mathcal{K})$}
\label{algo.hiddencore}
\DontPrintSemicolon
\BlankLine
\KwIn{the hidden graph $G$, the number $\mathcal{K}$}
\KwResult{the $\mathcal{K}$-core if exists, $\emptyset$ otherwise}
\BlankLine
$V_S \leftarrow \emptyset$ \; 
$n \leftarrow |V|$ \;  
$\mathcal{C} \leftarrow  n$ \; 
$mpd(G) \leftarrow  n-1$ \; 
\While{more vertices in $G$} {
   	\If{$mpd(G) < \mathcal{K}$} {
   		invoke \textsc{CoreDecomposition}($S$) \;
   	}
    $u \leftarrow$ vertex $u$ with $state(u)=0$ and $pd(u) \geq \mathcal{K}$ \;
    %\Comment {Criterion 1}
	\While{$s(u)+e(u)=n-1$} {
    	$V_S \leftarrow V_S \cup \{u\}$ \;
        $u\leftarrow$ vertex $u$ with $state(u)=0$ and $pd(u) \geq \mathcal{K}$ \;
    }    
    $mpd(G) \leftarrow 0$ \;
	\While{($u \leq $ last vertex in $G$)} {
        $v$ $\leftarrow$ select a destination vertex $v$, $pd(v) \geq \mathcal{K}$ \;
		$proberesult \leftarrow f(u,v)$ \;
        \If{(proberesult = empty)} {
        	$e(u)++$ \; 
            $e(v)++$ \com{/* empty edge found */} \; 
			$state(u)--$ \; 
            $state(v)--$ \;
            \If{\textsc{HiddenCoreCheck}$(u,\mathcal{K})$ = false} {
            	\Return $\emptyset$ \com{/* the $\mathcal{K}$-core does not exist */} \;
            }
            \If{\textsc{HiddenCoreCheck}$(v,\mathcal{K})$ = false} {
            	\Return $\emptyset$ \com{/* the $\mathcal{K}$-core does not exist */} \;           
            }
            %\While{(v is not NULL)} {
            	$u \leftarrow$ vertex with $state(u)=0$ and $pd(u) \geq \mathcal{K}$ \;
            	%$v \leftarrow$ select destination vertex \;
            %}
         }
        \Else {
        	$s(u)++$ \; 
            $s(v)++$ \com{/* solid edge found */} \; 
            \If{$s(v)+e(v) = n-1$} {
            	$V_S \leftarrow V_S \cup \{v\}$ \;
            }            
            \If{$s(u)+e(u) = n-1$} {
            	$V_S \leftarrow V_S \cup \{u\}$ \;
                $u\leftarrow$ next vertex with $state(u)=0$ and $pd(u) \geq \mathcal{K}$ \;
            }    
        }
    }         
}
\end{algorithm}

The pseudocode of the proposed technique is summarized in Algorithms~\ref{algo.hiddencorecheck} and~\ref{algo.hiddencore}. Early termination is possible at Lines~21-24, when \textsc{HiddenCoreCheck} fails to satisfy the requirements. If the execution reaches Line~7, the \textsc{CoreDecomposition} algorithm is invoked to discover the $\mathcal{K}$-core. At this stage, again we have two options: $i$) either the $\mathcal{K}$-core exists and the algorithm returns the corresponding subgraph, or $ii$) the $\mathcal{K}$-core does not exist, in which case the algorithm returns the empty set.

To demonstrate the way \textsc{HiddenCore} operates, in the sequel we provide a running example based on the small graph shown in Figure~\ref{fig.hiddencoreexample}. 
We provide two different cases with respect to the result: $i$) the requested $\mathcal{K}$-core does not exist (negative answer) and $ii$) the requested $\mathcal{K}$-core does exist (positive answer). Tables~\ref{tab.exampleK4_r1} to \ref{tab.exampleK5_r2} depict the actions taken and the status of each probe applied, depending on the outcome (\textit{solid} or \textit{empty}).

\begin{figure}[!ht]
\includegraphics[scale=0.5]{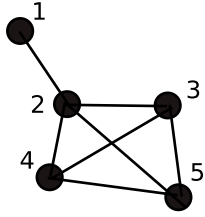}
\vspace{-0.2cm}
\caption{Example graph to illustrate the steps of \textsc{HiddenCore}. The maximum existing core is the 3-core, composed of the vertices 2, 3, 4 and 5.}
\label{fig.hiddencoreexample}
\end{figure}

Firstly, let us assume that the user is interested in the 4-core of the hidden graph $G$, and thus, $\mathcal{K}=4$. Table~\ref{tab.exampleK4_r1} shows the actions taken. In particular, after the fist \textit{empty} probe $f(1,3)$, the algorithm terminates because there are at least two vertices (1 and 3) with degree less than $\mathcal{K}$, and therefore, the $\mathcal{K}$-core does not exist in $G$.
Out of the 10 possible probes, only two of them were executed which is translated to 80\% gain with respect to the number of probes.

Secondly, we check the progress of \textsc{HiddenCore} for $\mathcal{K}=3$. Tables~\ref{tab.exampleK5_r1} and \ref{tab.exampleK5_r2} depict the actions taken for the 1st and 2nd round respectively. Note that, one round is not adequate for the algorithm to terminate, since the termination condition is not satisfied. Moreover, Table~\ref{tab.exampleK5_var} shows the values of the most important variables. In this case, out of the 10 possible probes, 9 of them are executed by \textsc{HiddenCore}, resulting in a 10\% gain.

\begin{table}[!ht]
\begin{center}
\caption{\textsc{HiddenCore} probes for $\mathcal{K}=4$. (Round 1)}
\vspace{-0.2cm}
\label{tab.exampleK4_r1}
\renewcommand{\arraystretch}{0.9}
\begin{tabular}{|l||l|c|c|}
\hline
{\bf Probe}			& {\bf Actions/Notes}\\ \hline\hline
$f(1,2)=solid$	& - $s(1)$++ and $s(2)$++ \\ \hline
$f(1,3)=empty$	& \makecell[l]{- $e(1)$++ and $e(3)$++\\ - $pd(1)=3$ and $pd(3)=3$ which is less than $\mathcal{K}$ \\ ~~~~and so they are eliminated.\\ - $\mathcal{C}=3$ is less than $\mathcal{K}+1$\\ 
- \textsc{HiddenCore} terminates and returns \textit{false} \\ }  \\ 
\hline
\end{tabular}
\end{center}
\end{table}

\begin{table}[!ht]
\begin{center}
\caption{\textsc{HiddenCore} probes for $\mathcal{K}=3$. (Round 1)}
\vspace{-0.2cm}
\label{tab.exampleK5_r1}
\renewcommand{\arraystretch}{0.9}
\begin{tabular}{|l||l|c|c|}
\hline
{\bf Probe}			& {\bf Actions/Notes}						\\ \hline\hline
$f(1,2)=solid$		& \makecell[l]{- $s(1)$++ and $s(2)$++}		\\ \hline
$f(1,3)=empty$		& \makecell[l]{- $e(1)$++ and $e(3)$++ 		\\
					- $state(1)--$, $state(3)--$ 				\\ 
                    - $pd(1)=3$ and $pd(3)=3$ 					\\ 
                    - continue }  								\\\hline
$f(2,3)=solid$	& \makecell[l]{- $s(2)$++ and $s(3)$++}			\\ \hline
$f(2,4)=solid$	& \makecell[l]{- $s(2)$++ and $s(4)$++}			\\ \hline
$f(2,5)=solid$	& \makecell[l]{- $s(2)$++ and $s(5)$++ 			\\ 
					- vertex 2 is inserted into $S$}			\\ \hline
$f(4,1)=empty$	& \makecell[l]{- vertex 3 cannot become a source \\ 
						~~~~because $state(3)<0$.\\
					- $e(4)$++ and $e(1)$++ \\ 
                	- $state(4)--$ and $state(1)--$ \\ 
                	- $pd(4)=4$ and $pd(1)=2$ ($<\mathcal{K}$)  \\
                	- vertex 1 is eliminated\\
                	- $\mathcal{C}=4$ \\
					- continue}\\ \hline
$f(5,3)=solid$	& \makecell[l]{- vertex 1 cannot be a destination (pruned)\\
					- vertex 2 cannot be a destination \\ 
                    ~~~~because the probe $f(2,5)$ has been used \\
					- vertex 3 is the next destination \\
					- $s(5)$ ++ and $s(3)$++} \\ \hline
$f(5,4)=solid$	& \makecell[l]{- $s(5)$++ and $s(4)$++ \\ 
					- no available probes for vertex 5 \\
                    - vertex 5 is inserted to $S$\\
					- end of Round 1} \\ \hline
\end{tabular}
\end{center}
\end{table}

\begin{table}[!ht]
\begin{center}
\caption{Vertex status for $\mathcal{K}=3$ after Round 1}
\label{tab.exampleK5_var}
\renewcommand{\arraystretch}{1}
\begin{tabular}{|c||l|l|l|l|c|}
\hline
{\bf Vertex $u$} & {\bf $s(u)$} & {\bf $e(u)$} & {\bf $state(u)$} & {\bf $pd(u)$} & {\bf status}\\ \hline \hline
{\sf 1}	& \makecell{1} & \makecell{2} & \makecell{-2} & \makecell{2} & eliminated \\ 
\hline
{\sf 2}	& \makecell{4} & \makecell{0} & \makecell{0} & \makecell{4} & in $S$ \\ 
\hline
{\sf 3}	& \makecell{2} & \makecell{1} & \makecell{-1} & \makecell{3} & in $G$ \\ 
\hline
{\sf 4}	& \makecell{2} & \makecell{1} & \makecell{-1} & \makecell{3} & in $G$ \\ 
\hline
{\sf 5}	& \makecell{3} & \makecell{0} & \makecell{0} & \makecell{3} & in $S$ \\ 
\hline
\end{tabular}
\end{center}
\end{table}

\begin{table}[!ht]
\begin{center}
\caption{\textsc{HiddenCore} probes for $\mathcal{K}=3$. (Round 2)}
\label{tab.exampleK5_r2}
\renewcommand{\arraystretch}{1.2}
\begin{tabular}{|l||l|c|c|}
\hline
{\bf Probe}			& {\bf Actions/Notes}\\ \hline\hline
$f(3,4)=solid$	& \makecell[l]{- $s(3)$++ and $s(4)$++ \\  
					- no more probes for vertices 3 and 4 \\
					- both are inserted to $S$ with degree 3 \\
					- all vertices with degree $\geq \mathcal{K}$ \\
					- have been detected we \\ 
					- invoke \textsc{CoreDecomposition} \\
					- the 3-core does exist \\
                    - \textsc{HiddenCore} returns a positive result} \\ \hline
\end{tabular}
\end{center}
\end{table}

\subsection{Runtime Cost and Complexity}

We conclude this part of the paper by discussing about the overall cost of the \textsc{HiddenCore} algorithm.
Note that, the runtime cost is mainly defined by the number of edge probing queries issued. Assuming that the cost of each probing query, i.e., the computation of the function $f()$, is significant, reducing the number of probes is essential. 

\begin{figure*}[!th]
\begin{center}
\begin{minipage}{5.6cm}
\centerline{\epsfig{file=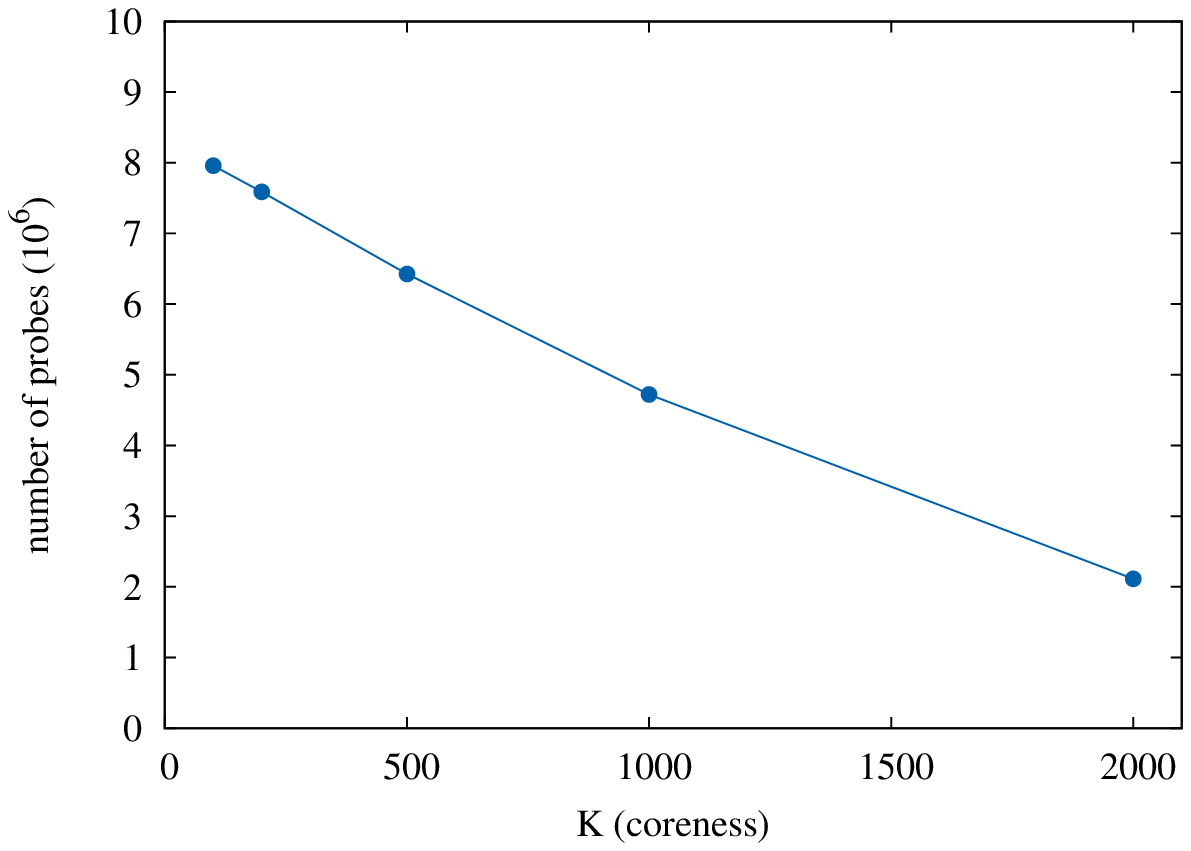,width=5.5cm}}
\centerline{}
\centerline{(a) {\sf ego-Facebook}}
\end{minipage} 
\begin{minipage}{5.6cm}
\centerline{\epsfig{file=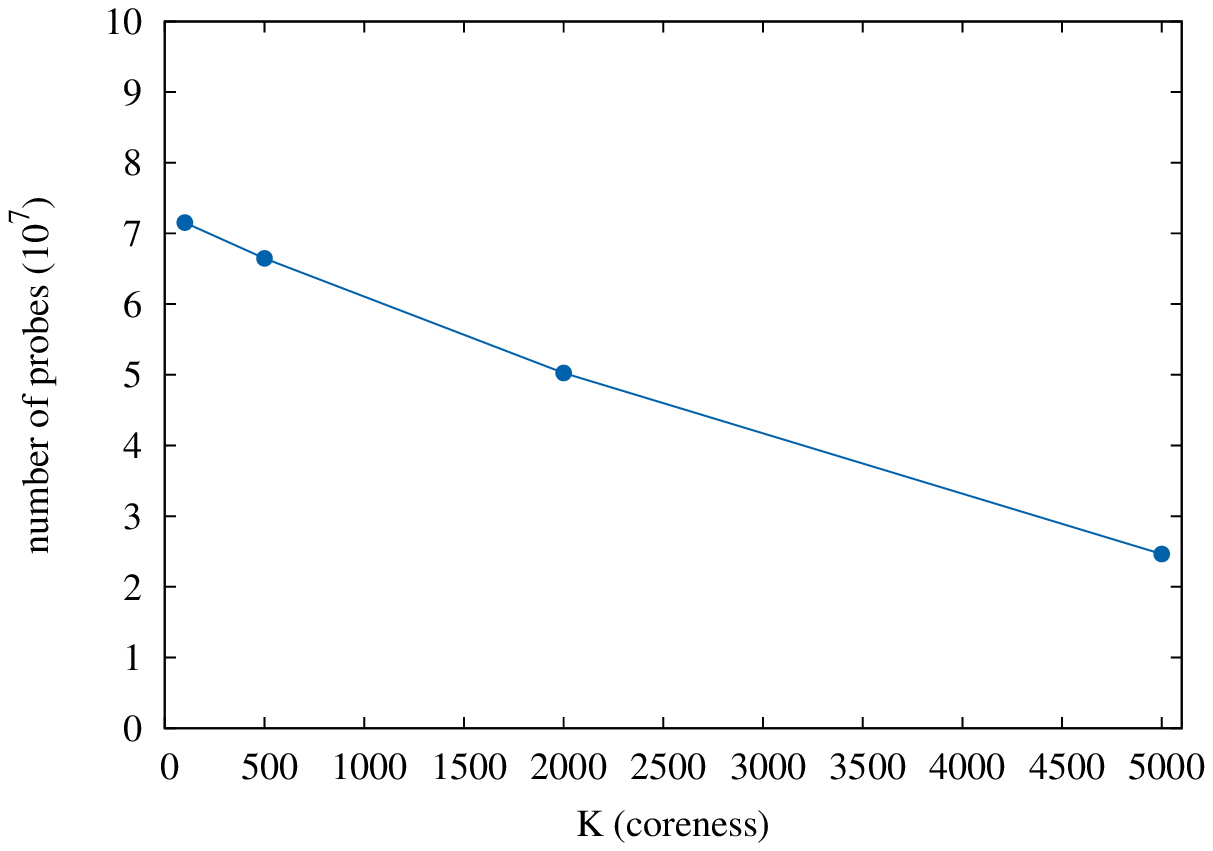,width=5.5cm}}
\centerline{}
\centerline{(b) {\sf ca-HepPh}}
\end{minipage}
\begin{minipage}{5.6cm}
\centerline{\epsfig{file=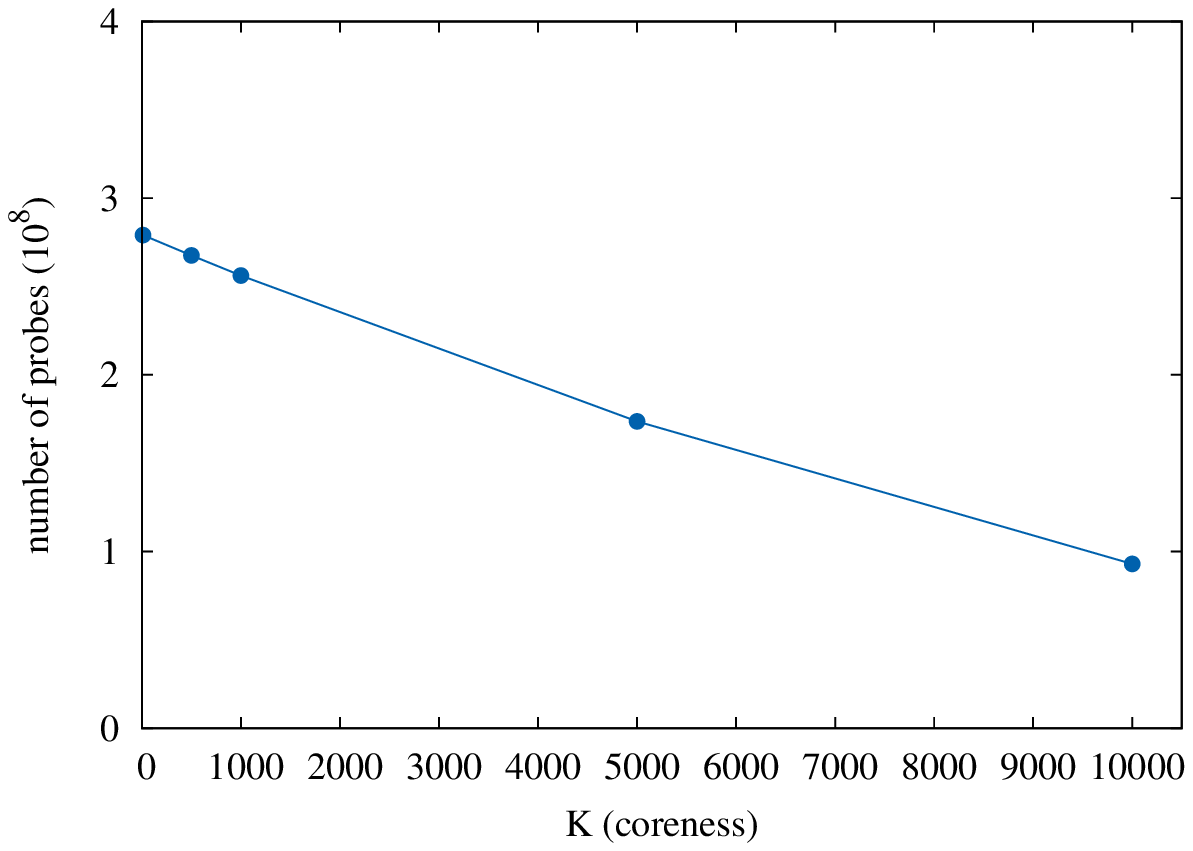,width=5.5cm}}
\centerline{}
\centerline{(c) {\sf soc-Gplus}}
\end{minipage}
\centerline{}
\centerline{}
\begin{minipage}{5.6cm}
\centerline{\epsfig{file=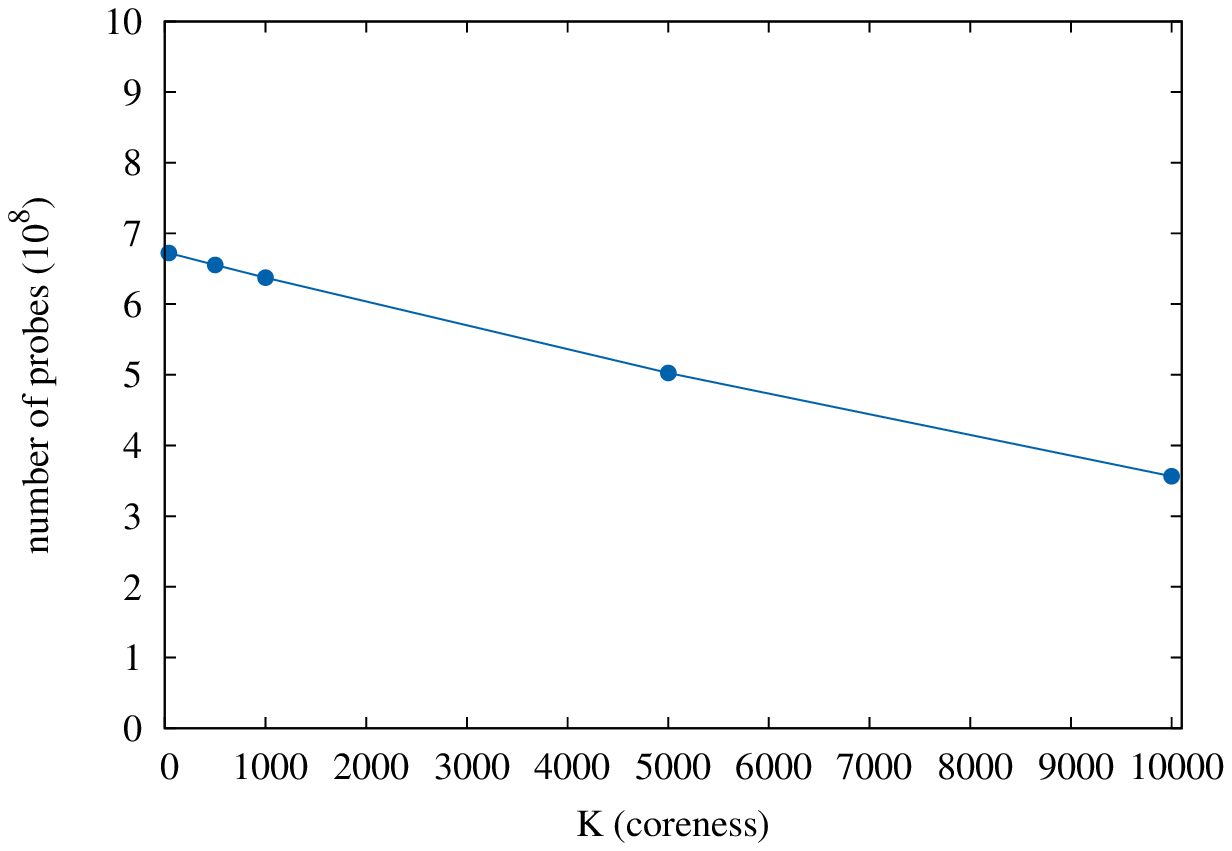,width=5.5cm}}
\centerline{}
\centerline{(d) {\sf email-Enron}}
\end{minipage}
\begin{minipage}{5.6cm}
\centerline{\epsfig{file=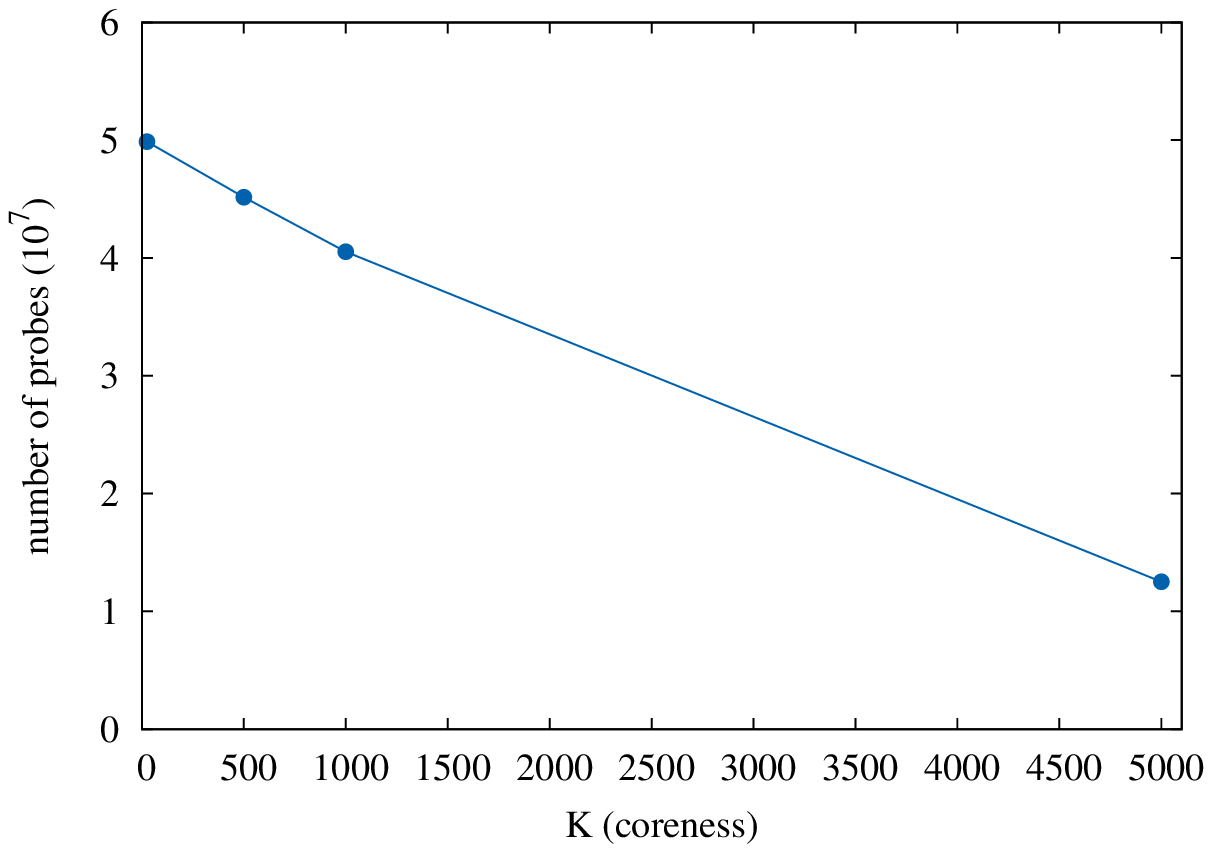,width=5.5cm}}
\centerline{}
\centerline{(e) {\sf power-law5K}}
\end{minipage}
\end{center}
\caption{Number of probes issued for different data sets varying the parameter $\mathcal{K}$.} 
\label{fig.resultsprobes}
\end{figure*}

The second factor that has a direct impact on the performance of the algorithm is the number of primitive operations performed. This may refer to the number of comparisons performed, or the number of searches in lookup tables, and many more. In general, the cost of a probe execution is several orders of magnitude more expensive than a primitive operation and one may think that the total runtime cost is defined by the number of probes. However, this is true because if the number of primitive operations increases significantly, the computational cost may increase significantly as well. For example, an algorithm $\mathcal{A}_1$ that requires 1000 probes and $\mathcal{O}(n \log n)$ primitive operations may be more efficient than another algorithm $\mathcal{A}_2$ which needs 10 probes and $\mathcal{O}(n^3)$ primitive operations. Therefore, it is essential to minimize the number of probes, as well as the number of primitive operations \textit{per probe}. 

Let $probes$ denote the total number of probes issued by the algorithm. Based on the previous discussion, each probe triggers a sequence of primitive operations that are at most $\mathcal{O}(\log n)$, resulting in a total complexity of $\mathcal{O}(probes \cdot \log n)$ devoted for updating the bookkeeping data structures. Moreover,  for each probe issued there is additional $\mathcal{O}(\log(\mathcal{K}))$ cost to update the minheap data structure that accommodates the $\mathcal{K}+1$ vertices with the highest number of solid edges. However, this cost does not change the $\mathcal{O}(probes \cdot \log n)$ complexity since always $\mathcal{K} \leq n$. 

To that, we need to also add the cost for running the \textsc{CoreDecomposition} algorithm which is in $\mathcal{O}(m+n)$. It is very interesting to provide lower bounds with respect to the number of probes that are required for the discovery of the $\mathcal{K}$-core, as a function of the number of vertices and other structural properties.

\smallExtraSpacing

\section{Performance Evaluation}
\label{sec.performance}

This section contains performance evaluation results, demonstrating the runtime costs of the GSOE algorithm both for detecting high degree vertices and for the discovery of the $k$-core of a hidden graph. All techniques are implemented in the C++ programming language. For the experiments, we have used real-world as well as synthetic graphs following a power-law degree distribution. The data sets used are summarized in Table~\ref{tab.datasets}.

\begin{table}[!h]
\begin{center}
\caption{Data sets used in the experimental evaluation.}
\label{tab.datasets}
\renewcommand{\arraystretch}{1.3}
\begin{tabular}{|l||c|c|}
\hline
{\bf Graph}			& \#vertices ($|V|)$	& \#edges ($|E|$)	\\\hline\hline
{\sf ego-Facebook}	& 4,039					& 88,234			\\ \hline
{\sf ca-HepPh}		& 12,008				& 118,521			\\ \hline
{\sf soc-Gplus}		& 23,600				& 39,200			\\ \hline
{\sf email-Enron}	& 36,692				& 183,831			\\ \hline
{\sf power-law1K}	& 1,000					& 50,000			\\ \hline
{\sf power-law2K}	& 2,000					& 100,000			\\ \hline
{\sf power-law3K}	& 3,000					& 150,000			\\ \hline
{\sf power-law5K}	& 5,000					& 250,000			\\ \hline
\end{tabular}
\end{center}
\end{table}

The real-world graphs have been downloaded from the SNAP repository at Stanford (\textsf{http://snap.stanford.edu}) and the Network repository (\textsf{http://networkrepository.com}). For the synthetic graphs we have used the \textsf{GenGraph} tool, which implements the graph generation algorithm described in~\cite{VL05}. In particular, \textsf{GenGraph} generates a set of $n$ integers in the interval $[d_{min},d_{max}]$ obeying a power-law distribution with exponent $\alpha$. These integers are used as the degree sequence and they define the degrees of the vertices of the synthetic graph that is produced. 

\begin{table*}[!t]
\begin{center}
\caption{The percentage gain of the number of probes performed by \textsc{HiddenCore} in comparison to brute force.}
\label{tab.resultsgain}
\renewcommand{\arraystretch}{1.3}
\begin{tabular}{|c||c|c|c|c|c|}
\hline
{$\mathcal{K}$}	& ~{\sf ego-Facebook}~	& ~{\sf ca-HepPh}~	& ~{\sf soc-Gplus}~	& ~{\sf email-Enron}~	& ~{\sf power-law5K}~	\\ \hline\hline
10				&			&			& 0.4\%			&			&				\\ \hline
25				&			&			&				&			& 0.2\%			\\ \hline
40				& 			&			&				& 0.2\%		&				\\ \hline
100				& 2.4\%		& 0.8\%		&				&			&				\\ \hline
200				& 7\%		&			&				&			&				\\ \hline
500				& 21.2\%	& 8\%		& 4\%			& 2.5\%		& 9.9\%			\\ \hline
1,000			& 42.1\%	&			& 8.2\%			& 5\%		& 19\%			\\ \hline
2,000			& 74\%		& 30.3\%	&				&			& 				\\ \hline
5,000			&			& 66\%		& 38\%			& 25\%		&				\\ \hline
10,000			&			&			& 68\%			& 47\%		&				\\ \hline
\end{tabular}
\end{center}
\end{table*}

Figure~\ref{fig.resultsprobes} depicts the performance of the \textsc{HiddenCore} algorithm for all available data sets. In particular, we monitor the total number of probes vs. the parameter $\mathcal{K}$, which defines the order of the requested core. We observe that as $\mathcal{K}$ increases, the total number of probes decreases significantly. By using higher $\mathcal{K}$ values, we are requesting cores that contain more vertices (at least $\mathcal{K}+1$) with higher degree (at least $\mathcal{K}$). Therefore, the early termination conditions of the \textsc{HiddenCore} algorithm have more chances to be fulfilled resulting in better performance than the brute force algorithm which requires all $\mathcal{O}(n^2)$ probes to be executed first. 

\begin{figure}[!b]
\begin{center}
\begin{minipage}{5.6cm}
\centerline{\epsfig{file=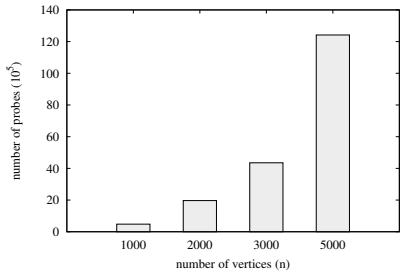,width=5.5cm}}
\centerline{}
\centerline{(a) $\mathcal{K}=50$}
\end{minipage}
\begin{minipage}{5.6cm}
\centerline{\epsfig{file=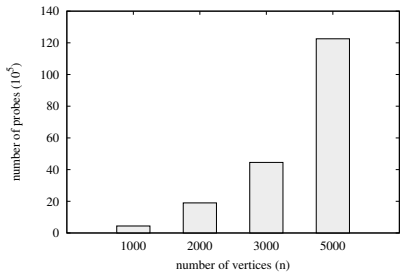,width=5.5cm}}
\centerline{}
\centerline{(b) $\mathcal{K}=100$}
\end{minipage}
\end{center}
\caption{Number of probes issued for different number of vertices in synthetic power-law graphs ({\sf power-law1K}, {\sf power-law2K}, {\sf power-law3K}, {\sf power-law5K}).} 
\label{fig.resultsprobespl}
\end{figure}

On the other hand, as we reduce $\mathcal{K}$, more probes are required in order to rank the appropriate vertices. This is due to the fact that more vertices survive the constraints and therefore, more probes will be required to completely determine their degree, before the invocation of \textsc{CoreDecomposition}. This leads to an increase in the total number of probes.

Table~\ref{tab.resultsgain} shows the percentage gain on the total number of probes issued, for different values of $\mathcal{K}$ and different graphs. As expected, for small $\mathcal{K}$ values a significant number of probes is performed. As $\mathcal{K}$ increases, more probes are saved.

Finally, in Figure~\ref{fig.resultsprobespl}, we depict the number of probes issued vs. the size of the synthetic power-law graph, for two different values of $\mathcal{K}$ (50 and 100). Note that, as the number of vertices increases, the number of edges increases too, as shown in Table~\ref{tab.datasets}. We observe that for the same value of $\mathcal{K}$, the number of probes also increases rapidly, showing a quadratic rate of growth. This behavior is explained by the fact that the maximum number of probes also grows in a quadratic rate, since for $n$ vertices the maximum number of probes equals $n(n-1)/2$.

By observing the experimental results we conclude that probe savings is extremely hard for small values of $\mathcal{K}$, provided that we need a 100\% accurate answer regarding the existence of the $\mathcal{K}$-core. We believe that there is still room for improvements towards reducing the number of probes further. Moreover, it turns out that detecting the largest value of $\mathcal{K}$ for which the $\mathcal{K}$-core exists is an even more challenging problem.

\section{Conclusions and Future Work}
\label{sec.conclusions}

Hidden graphs are extremely flexible structures since they can represent associations between entities without storing the edges explicitly. This way, many different relationship types can be described, since the only change is the function $f(u,v)$ that should be invoked to reveal the existence of an edge $(u,v)$. 

In such a setting, existing graph algorithms cannot be applied directly, since the set of neighbors for each node is not known in advance. Since edge probing queries may be extremely expensive to execute (i.e., may involve running complex algorithmic techniques), the aim is to minimize their number as much as possible, to guarantee efficient execution.

In this work, we have studied the problem of core detection in hidden graphs: given a hidden graph $G(V,f())$ and an integer $\mathcal{K}$, detect the $\mathcal{K}$-core of $G$, or return \textit{false} if such a core does not exist. In general, the core decomposition problem in conventional graphs (i.e., graphs with a known set of edges) can be solved in linear time $\mathcal{O}(m+n)$, where $n$ is the number of nodes and $m$ is the number of edges. However, to be able to apply the linear algorithm, the set of edges must be known in advance, meaning that $\mathcal{O}(n^2)$ probes must be executed first, which is extremely costly. 

We have shown that by using a generalization of the Switch-On-Empty (SOE) algorithm (GSOE) together with a proper bookkeeping strategy for the probes performed, we can execute efficiently the following tasks: $i$) compute the $k$ nodes with the highest degrees, and $ii$) detect the presence or absence of a $\mathcal{K}$-core in the hidden graph by performing significantly less probes. We note that this is the first work that attacks the core discovery problem in hidden graphs. The proposed techniques are useful in hidden network exploration and visualization. Moreover, 
the generalization of the SOE algorithm as well as the bookkeeping techniques applied for the design of \textsc{HiddenCore} can be used to solve other related problems in the area. We highlight the following future research directions: ~\\

\begin{itemize}
\item
In some cases, we are interested in the core number of specific nodes. Local computations are required in this case, since we are not interested in the core numbers of all nodes. It is challenging to combine the concept of the hidden graph with local computation in this case, in order to minimize the number of probes.
\item
The concept of the densest subgraph is strongly related to that of core decomposition, since on of the $k$-cores of a graph is a $\frac{1}{2}$-approximation of the densest subgraph. Detecting dense subgraphs is considered a very interesting problem in the hidden graph context, especially if additional user-based constraints are used (e.g., each dense subgraph must contain at least $\alpha$ edges).
\item
In some cases, we just need the maximum core of the hidden graph. Spotting the maximum core is very challenging since initially we have no available information about the degrees of the vertices. A potential solution to this problem, is to provide an \textit{incremental} version of \textsc{HiddenCore} in order to apply the algorithm continuously (e.g., a logarithmic number of times) until we spot the maximum core. 
\item
The algorithms covered in this work, have been designed using a centralized point of view. However, assuming that in certain cases probes could be performed in parallel, it is interesting to investigate parallel algorithms towards reducing the overall runtime by exploiting multiple resources.
\item
The techniques covered in the paper have been developed towards a deterministic and exact approach, meaning that the $\mathcal{K}$-core of the hidden graph (if it exists) it is computed accurately. Another possible approach is to adopt a randomized perspective, providing probabilistic guarantees about the correctness of the algorithm by reducing significantly the number of probes applied. 
\end{itemize}
~\\

\moreSpacing

\bibliographystyle{ACM-Reference-Format}
\bibliography{hcore}

\clearpage

\end{document}